\documentclass[conference,10pt]{IEEEtran}
\usepackage{cite}

\usepackage{graphicx}
\usepackage{ifpdf}
\ifpdf
  \usepackage{epstopdf}
\fi
\usepackage[cmex10]{amsmath}
\usepackage{amsthm}
\usepackage{amssymb}
\usepackage{cases}
\usepackage{bm}

\usepackage[caption=false,font=footnotesize]{subfig}
\usepackage{fixltx2e}
\usepackage{url}

\begin{document}
%
\title{On the Statistical Multiplexing Gain of\\Virtual Base Station Pools}
%
%
%

\author{\IEEEauthorblockN{Jingchu Liu, Sheng Zhou, Jie Gong, Zhisheng Niu}
        \IEEEauthorblockA{Tsinghua National Laboratory for Information Science and Technology\\
        Department of Electronic Engineering, Tsinghua University\\
        Beijing 100084, China \\
        Email: liu-jc12@mails.tsinghua.edu.cn,  sheng.zhou@tsinghua.edu.cn\\ gongj13@mail.tsinghua.edu.cn, niuzhs@tsinghua.edu.cn\\}
        \and
        \IEEEauthorblockN{Shugong Xu}
        \IEEEauthorblockA{Intel Labs\\
        \\
        Beijing 100080, China\\
        Email: shugong.xu@intel.com}
        }

%
%



\maketitle

\newtheorem{kcr}{Theorem}
\newtheorem{rev}[kcr]{Theorem}
\newtheorem{lpl}[kcr]{Theorem}

\begin{abstract}
Facing the explosion of mobile data traffic, cloud radio access network {(C-RAN)} is proposed recently to overcome the efficiency and flexibility problems with the traditional {RAN} architecture by centralizing baseband processing. However, there lacks a mathematical model to analyze the statistical multiplexing gain from the pooling of virtual base stations (VBSs) so that the expenditure on fronthaul networks can be justified. In this paper, we address this problem by capturing the session-level dynamics of {VBS} pools with a multi-dimensional Markov model. This model reflects the constraints imposed by both radio resources and computational resources. To evaluate the pooling gain, we derive a product-form solution for the stationary distribution and give a recursive method to calculate the blocking probabilities. For comparison, we also derive the limit of resource utilization ratio as the pool size approaches infinity. Numerical results show that {VBS} pools can obtain considerable pooling gain readily at medium size, but the convergence to large pool limit is slow because of the quickly diminishing marginal pooling gain. We also find that parameters such as traffic load and desired Quality of Service {(QoS)} have significant influence on the performance of {VBS} pools.
\end{abstract}


%
\IEEEpeerreviewmaketitle

\section{Introduction}
\IEEEPARstart{I}{n} recent years, the proliferation of mobile devices such as smart phones and tablets, together with the applications enabled by mobile Internet, has triggered the exponential growth of mobile data traffic \cite{vni}. To accommodate the rapid traffic growth, cellular networks have been continuously evolving with smaller cell size, wider bandwidth, and more advanced transmission technologies. However, the problems that arise, such as the increased interference and operational costs, are difficult to solve with the traditional RAN architecture, in which the communication-related functionalities are packed into stand-alone base stations (BSs) and the cooperation between {BSs} is limited by the backhaul network.

To overcome the shortcomings of the traditional RAN architecture, cloud radio access network {(C-RAN)} \cite{cran} is proposed with centralized baseband processing. {C-RAN} can facilitate the adoption of cooperative signal processing and potentially reduce the operational costs. A similar idea is also proposed in \cite{wnc} under the name of wireless network cloud {(WNC)}. This kind of novel architectures has attracted substantial attentions. The key functionalities of {C-RAN} are investigated and its major use cases are identified in \cite{ngmn}. Centralized processing is also utilized in conjunction with dynamical fronthaul network switching to address the mobility and energy efficiency issues of small cell scenarios in \cite{colony,fluidnet}. Concerning about realization, it is demonstrated in \cite{vbs,cloudiq,bigstation} that the functions of base band units (BBUs) can be implemented in the form of virtual base station (VBS) software that runs on general-purpose-platform (GPP) servers. Compared with the implementations based on dedicated hardware and software, GPP-based implementation provides more flexibility in the deployment of new functionalities and the provisioning of computational resources.

A {VBS} pool can be constructed by consolidating multiple VBSs onto GPP servers. This can be accomplished through running VBS instances as multiple threads in the same operating system ({OS}) or running them in separate real-time virtual machines ({VMs}) through virtualization technology. VBSs that are consolidated in this way can share the computational capacity of a single server or a cluster of servers, depending on the implementation. VBS pooling can improve the utilization ratio of computational resources so that related costs can be reduced.

Despite the numerous advantages mentioned above, the massive bandwidth requirement of {C-RAN's} fronthaul network poses a serious challenge. It can be estimated that transmitting the baseband sample of a single $20$MHz {LTE} antenna-carrier ({AxC}) requires around 1Gbps link bandwidth \cite{cpri}. Therefore a {C-RAN} with large-scale centralization may incur enormous fronthaul expenditure and potentially cancel out the gains from pooling {VBSs}. Fortunately, it is observed in \cite{cloudiq} that substantial statistical multiplexing gain can be obtained even with small scale centralization. Yet these observations are obtained from simulations based on estimated data, and a mathematical model is also needed to provide a general guideline for the design of realistic {VBS} pools. To this end, a model for {VBS} pools are proposed in \cite{gomez13} under the assumption of dynamic resource management, in which the amount compuatational capacity allocated to each baseband task is recalculated each time a new baseband task arrive. However, dynamic resource management may incur prohibitive overhead in realistic systems due to the stringent timeliness of baseband processing \cite{cloudiq}. Hence, semi-dynamic resource management, in which resource management algorithms run on much larger time scales than the processing of user sessions, may be more realistic. In addition, the interaction between the computational and radio resources is also not addressed in existing works.

In this article, we take a novel approach to analyze the statistical multiplexing gain of VBS pools. Under the assumption of semi-dynamic resource management, we model the session-level dynamics of {VBS} pools with a continuous-time multi-dimensional Markov model constrained by both radio and computational resources. Thanks to the special structure and the reversibility of the proposed model, we derive a product-form expression for its stationary distribution and give a recursive method for computing the user session blocking probabilities. To compare with numerical results, we further derive the limit of computational resource utilization ratio when there are infinite VBSs in the pool. Numerical results show that system parameters including pool size, traffic load, and the quality of service ({QoS}) have significant influence on the performance of {VBS} pool, which provides important implications for realistic system design.

The rest of the paper is organized as follows. Section \ref{sec:model} introduces the proposed model and provides proof for its reversibility. Section \ref{sec:solution} derives the product-form stationary distribution, the expression for blocking probabilities, and the large pool limit of resource utilization ratio. In Section \ref{recur} we give a recursive method for computing session blocking probabilities under arbitrary parameters. Section \ref{sec:numerical} presents the numerical results and discusses their implications on realistic system design. And the paper is concluded in section \ref{sec:conclusion}.

\section{A Multi-Dlimentional Markov Model} \label{sec:model}
In this section, we introduce the proposed model and provide proof for its reversibility. We model a {VBS} pool with $M$ {VBSs}, which share a total of $N$ computational servers. Each {VBSs} is connected to a remote radio unit (RRU) and equipped with $K$ units of radio resources. For simplicity, hereafter we refer to radio and computational resources as r-servers and c-servers, respectively.
    \subsection{Arrival, Service, and Blocking}
    \label{subsec:arrival}
    User sessions arrive independently in the coverage area of these {VBSs} following identical independent Poisson processes with arrival rate $\lambda$, and are served independently with exponential service time with mean $\mu^{-1}$. We assume exponential service time basing on the assumption that the length of users' data queue are i.i.d exponentially distributed. Defining the number of sessions in the $m$-th {VBS} to be $k_{m}$, then the number of sessions in all the pooled {VBSs} can be described with an $M$-dimensional vector
    $\bm{k} = (k_1,\cdots,k_{m},\cdots,k_{M})^T.$
    Given the Markovian property of the arrival and service of user sessions, it is obvious that $\bm{k}$ is a continuous-time $M$-dimensional Markov chain.

    Each active user session simultaneously occupies a r-server and a c-server, and releases them after being served. When a user session arrives, the pool scheduler will monitor the number of r-servers and c-servers to decide whether or not to accept the session. The session is accepted when the number of r-servers in the serving {VBS} is less than $K$ and the number of c-servers in the pool is less than $N$. Otherwise the session is rejected by the scheduler. This blocking policy reflects the constraint by both radio and computational resources. By reserving enough radio and computational resources for active sessions, we can guarantee the {QoS} for active sessions will not degrade under fluctuating processing load. The {QoS} for all sessions is reflected by the overal blocking probability. Note although different VBSs may have different {QoS} requirements, we assume VBSs has the same blocking probability threshold for the simplicity of analysis. Taking the blocking policy into consideration, we get the set of possible states:
        \begin{equation}
        \label{stateset}
        \begin{aligned}
            \mathbb{K} = \{& \bm{k} \mid 0 \le k_1,\cdots,k_M \le K,   &0 \le \sum_{m=1}^{M} k_m \le N \}
        \end{aligned}
        \end{equation}

    \subsection{Transition Rates}
    The state of $\bm{k}$ changes as user sessions arrive and depart. With the assumptions we made for session arrival and service in \ref{subsec:arrival}, there can only be a single session arrival or departure at any epoch. Thus, only a single entry of $\bm{k}$ can change at any epoch, and the change is either $+1$ or $-1$. In other words, $\bm{k}$ is a multi-dimensional birth-and-death process. Then we get the transition rate from state $\bm{k^{(i)}}$ to state $\bm{k^{(j)}}$ as:
        \begin{equation}
        \label{transition}
        q_{\bm{k^(i)} \bm{k^{(j)}}} = \begin{cases}
        \lambda,         & \text{if $\bm{k^{(j)}}-\bm{k^{(i)}}= \bm{e_m}$}\\
        k_m^i \mu, & \text{if $\bm{k^{(j)}}-\bm{k^{(i)}}= -\bm{e_m}$}\\
        0,               & \text{otherwise}
        \end{cases}
        \end{equation}
    where states $\bm{k^{(i)}}, \bm{k^{(j)}} \in \mathbb{K}$, $k_m^{(i)}$ is the $m$-th entry of $\bm{k^{(i)}}$, and $\bm{e_m}^{T} = (0,\cdots,0,\underbrace{1}_{m\text{-th}},0,\cdots,0)$.

    For the ease of understanding, Fig. \ref{2d} illustrates the state transition graph of a simple example with $M=2$, $K=3$, and $N=4$.

    \begin{figure}[!t]
    \centering
    \includegraphics[width=2.8in]{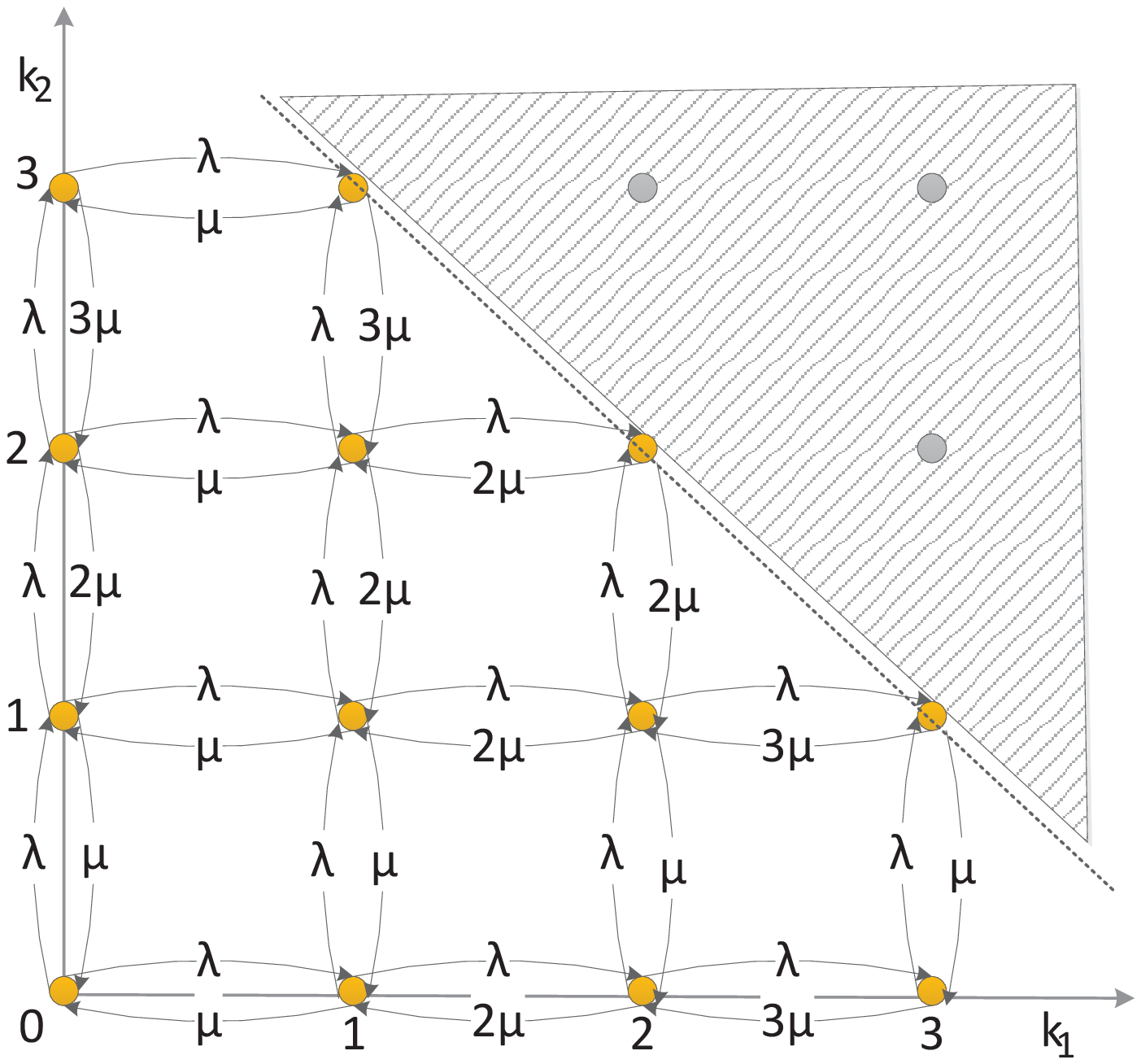}
    \caption{Transition graph of a pool with 2 {VBSs}. $K=3$, $N=4$.}
    \label{2d}
    \end{figure}

\section{Solution of The Proposed Model}
\label{sec:solution}
    \subsection{Stationary Distribution}
    A good property of $\bm{k}$ is the reversibility, since it can simplify the expression of the stationary distribution. The proof of reversibility proof can be found in \cite{kaufman81} for general cases, and is not repeated here.
    
    Since $\bm{k}$ is reversible, the local balance equation
    \begin{equation}\label{local}
    P(\bm{k^{(i)}}) \cdot q_{\bm{k^{(i)}}\bm{k^{(j)}}} = P(\bm{k^{(j)}}) \cdot q_{\bm{k^{(j)}}\bm{k^{(i)}}}
    \end{equation}
    holds for the statistical equilibrium of $\bm{k}$. Without loss of generality, let
    $$\bm{k^{(i)}} = (k_1,\cdots, k_m,\cdots,k_M)^T$$
    $$\bm{k^{(j)}} = (k_1,\cdots, k_m+1,\cdots,k_M)^T,$$
    and substitute (\ref{transition}) into (\ref{local}), we can get
    \begin{equation}\label{local2}
    \begin{aligned}
    & P(k_1,\cdots, k_m,\cdots,k_M) \cdot \lambda \\
    & = P(k_1,\cdots, k_m+1,\cdots,k_M) \cdot \left( k_m+1 \right) \mu
    \end{aligned}
    \end{equation}
    After a simple manipulation on (\ref{local2}), we have
    \begin{equation}\label{local3}
    \frac{P(k_1,\cdots, k_m+1,\cdots,k_M)}{P(k_1,\cdots, k_m,\cdots,k_M)}
    = \frac{a}{\left( k_m+1 \right)}
    \end{equation}
    where $a = \lambda / \mu.$ Clearly, this is a recursive equation for computing the stationary distribution. Continuing the recursion down to $0$ for the $m$-th entry, we can get
    \begin{equation}
    \begin{aligned}
    &P(k_1,\cdots, k_m,\cdots,k_M)\\
    &= P(k_1,\cdots,0,\cdots,k_M) \cdot \frac{a^{k_m}}{k_{m}!}
    \end{aligned}
    \end{equation}
    Then repeat the same process for other entries and we can get the expression for the stationary distribution of $\bm{k}$:
    \begin{equation}\label{static}
    P(\bm{k}) = P_0 \cdot \prod_{m=1}^{M}\frac{a^{k_m}}{k_m!}
    \end{equation}
    in which
    \begin{equation}
    P_0 = P(0,\cdots,0,\cdots,0)
    = \left( \sum_{\bm{k} \in \mathbb{K}}
    \prod_{m=1}^{M}\frac{a^{k_m}}{k_m!}\right)^{-1}
    \end{equation}
     can be derived from the fact that $$\sum_{\bm{k} \in \mathbb{K}}P(k_1,\cdots, k_m,\cdots,k_M) = 1$$

    \subsection{Blocking Probabilities}
    In our model, it is important that we can evaluate the blocking probability under arbitrary parameters. For ease of analysis, we first decompose the blocking events into two sets. We define the blocking events that are solely due to insufficient r-servers (i.e. $k_m^{-} = K, \sum_{i=1}^{M}k_i^{-}<N$) to be radio blocking, the blocking events that are due to insufficient c-servers (i.e. $\sum_{i=1}^{M}k_i^{-}=N$) to be computational blocking, and the union set of radio and computational blocking events to be overall blocking. Note that radio blocking and computational blocking are mutually exclusive. With above definition, we have overall blocking probability
    \begin{equation}
    P_{b} = P_{br} + P_{bc}
    \end{equation}
    with the probability of radio blocking
    \begin{equation}\label{pbr}
    \begin{aligned}
    P_{br}(N,M)
    &= \begin{cases}
    \sum\limits_{m=1}^{M}\frac{1}{M} \sum\limits_{\bm{k} \in \mathbb{K}_{<N}^{m,K} }P(\bm{k}),&N > K \\
    0 , & N \le K
    \end{cases}\\
    &= \begin{cases}
    \frac{P_0}{M} \sum\limits_{m=1}^{M}\sum\limits_{\bm{k} \in \mathbb{K}_{<N}^{m,K} }\prod\limits_{m=1}^{M}\frac{a^{k_m}}{k_m!},&N > K \\
    0 , & N \le K
    \end{cases}
    \end{aligned}
    \end{equation}
    and the probability of computational blocking
    \begin{equation}\label{pbc}
    \begin{aligned}
    P_{bc}(N,M)
    &= \sum_{\bm{k} \in \mathbb{K}_{=N}}P(\bm{k})\\
    &= P_0 \cdot \sum_{\bm{k} \in \mathbb{K}_{=N}}\prod_{m=1}^{M}\frac{a^{k_m}}{k_m!}
    \end{aligned}
    \end{equation}
    where
    $$\mathbb{K}_{=N} = \{ \bm{k} | k_1 +\cdots,k_M=N  \}$$
    $$\mathbb{K}_{<N}^{m,K} = \{ \mathbb{K}^{m,K} \cap \mathbb{K}_{<N}\}$$
    $$\mathbb{K}_{<N} = \{ \bm{k} | k_1 +\cdots,k_M<N  \}$$
    $$\mathbb{K}^{m,K} = \{ \bm{k} | k_m = K\}.$$
    Note that (\ref{static}) is symmetric in any two of its arguments:
    \begin{equation}
    \label{symm}
    P(\cdots, k_{i},\cdots,k_{j},\cdots) = P(\cdots, k_{j},\cdots,k_{i},\cdots),
    \end{equation}
    then (\ref{pbr}) can be simplified for $N > K$ as:
    \begin{equation}\label{pbcs}
    \begin{aligned}
    P_{br}(N,M)
    &= \frac{M}{M}\sum\limits_{\bm{k} \in \mathbb{K}_{<N}^{1,K}}P(\bm{k})\\
    &= P_0 \cdot \sum\limits_{\bm{k} \in \mathbb{K}_{<N}^{1,K}} \prod_{m=1}^{M}\frac{a^{k_m}}{k_m!}\\
    &= P_0 \cdot \frac{a^K}{K!} \cdot \sum\limits_{\bm{k} \in \mathbb{K}_{<N}^{1,K}}\prod_{m=2}^{M}\frac{a^{k_m}}{k_m!}
    \end{aligned}
    \end{equation}

    \subsection{Large Pool Limit}\label{subsec:large}
    As the number of pooled {VBSs} $M \to \infty$, the computational resource utilization ratio will approach a limit, which is described in the following theorem.
    \begin{lpl}[Large Pool Limit]
        The asymptotic utilization ratio of computational resources $\eta^*$ when $N = MK$ and $M \to \infty$ is bounded with
            \begin{equation}
            \frac{a \left( 1-P_b^{th}\right)}{K} \le \eta^* \le  \frac{a}{K}
            \end{equation}
            where $P_b^{th}$ is the desired threshold for total blocking probability.
    \end{lpl}

    \begin{proof}
        When c-servers are sufficiently-provisioned, i.e. $N = MK$, the entries of $\bm{k}$ are i.i.d. According to the law of large numbers, for any positive number $\epsilon$
            \begin{equation}\label{p}
            \lim_{M \to \infty} P(\left| \frac{\sum_{m=1}^{M}k_m}{M} - E\left[k_m\right] \right| > \epsilon) = 0
            \end{equation}
        In other word, the average number of sessions across all the {VBSs} in the pool will approach the average number of sessions in any of these {VBSs}, which can be calculated as
            \begin{equation}
            \begin{aligned}
            E\left[k_m\right]
            = \frac{\sum\limits_{i=0}^{K} i \cdot \frac{a^i}{i!}}{\sum\limits_{i=0}^{K}\frac{a^i}{i!}} =
            a \frac{\sum\limits_{i=0}^{K-1} \frac{a^i}{i!}}{\sum\limits_{i=0}^{K}\frac{a^i}{i!}}
            =a \cdot \left( 1- \frac{a^K}{K!}\left( \sum\limits_{i=0}^{K}\frac{a^i}{i!}\right)^{-1}\right)
            \end{aligned}
            \end{equation}

        Now that c-servers are sufficiently-provisioned, there can only be radio blocking. Hence, overall blocking probability equals to radio blocking probability. Also notice that
        $\frac{a^K}{K!}\left( \sum\limits_{i=0}^{K}\frac{a^i}{i!}\right)^{-1}$
        is the radio blocking probability of any of the {VBSs}, we have
            \begin{equation}\label{pb}
            \frac{a^K}{K!}\left( \sum\limits_{i=0}^{K}\frac{a^i}{i!}\right)^{-1} \le  P_b^{th}.
            \end{equation}

        The average number of sessions across all the {VBSs} can then be bounded with
        \begin{equation}
        a \cdot \left( 1-P_b^{th}\right) \le E\left[k_m\right] \le a.
        \end{equation}

        Combining (\ref{p}) and (\ref{pb}) and we get the utilization of the {VBS} pool when $M$ approaches infinity
        \begin{equation}
        \eta^* = \frac{M \cdot E\left[k_m\right]}{M K} = \frac{E\left[k_m\right]}{K} \in \frac{1}{K} \left[a \left( 1-P_b^{th}\right), a \right]
        \end{equation}
    \end{proof}

    The reason that we only investigate the case in which $N = MK$ is that, only the large pool limit when $N = MK$ truly reflects the potential of pool gain. Any $N > MK$ will over-provision c-servers and will not increase the actual pooling gain, while any $N < MK$ will exploit the potential of pooling gain to some extent. This theorem indicates that when the number of {VBSs} in the pool becomes very large, $(1-\eta^*)MK$ c-servers are idle for almost all the time.

    \section{Method for Recursive Evaluation} \label{recur}
    For the ease of numerical evaluation, we next give a recursive method for the computation of $P_{br}$ and $P_{bc}$. First, we define two auxiliary functions:
    \begin{equation}\label{cut}
    C(N,M) = \sum\limits_{\bm{k^i} \in \mathbb{K}_{=N}}\prod_{m=1}^{M}\frac{a^{k_m^i}}{k_m^i!}
    \end{equation}
    \begin{equation}\label{remain}
    R(N,M) = \sum\limits_{\bm{k^i} \in \mathbb{K}_{<N}}\prod_{m=1}^{M}\frac{a^{k_m^i}}{k_m^i!}.
    \end{equation}

    With (\ref{cut}) and (\ref{remain}), (\ref{pbr}) and (\ref{pbcs}) can be expressed as
    \begin{equation}
    \begin{aligned}
    P_{br}(N,M)
    &= \begin{cases}
     \frac{P_0 \cdot a^{K}}{K!}  R(N-K,M-1), &N > K\\
    0, & N \le K
    \end{cases}
    \end{aligned}
    \end{equation}

    \begin{equation}
    P_{bc}(N,M) = P_0 \cdot C(N,M)
    \end{equation}

    with
    \begin{equation}
    P_0 = \left( R(N + 1, M)\right)^{-1}.
    \end{equation}

    For the computation of $C(N,M)$ and $R(N,M)$, we can find the following recursive relationships:
    \begin{equation}
    \begin{aligned}
    C(N,M)
    &= \begin{cases}
    \frac{a^{N}}{N!}, & M=1 \\
    \begin{aligned}
    \sum\limits_{i=N_1}^{N_2}\frac{a^{i}}{i!}C(N-i,M-1), \end{aligned} & M>1
    \end{cases}
    \end{aligned}
    \end{equation}

    \begin{equation}
    \begin{aligned}
    R(N,M)
    & = \begin{cases}
    0, & N = 1 \\
    \begin{aligned}&R(N+1,M) \\&{ }-C(N  ,M),\end{aligned}   & 2 \le N \le MK\\
    \begin{aligned}&R(N-1,M) \\&{ }+C(N-1,M),\end{aligned} & 2 \le N \le MK\\
    \left( \sum_{i=1}^{K}\frac{a^i}{i!} \right)^M, & N = MK+1.
    \end{cases}
    \end{aligned}
    \end{equation}
    where $N_1 = \max{(0,N-(M-1)K)}$ and $N_2 = \min{(K,N)}$. These relationships can speed-up the computation of $C(N,M)$ and $R(N,M)$ by re-utilizing cached results.

    \section{Numerical Results and Discussions}\label{sec:numerical}
    In this section, we give the numerical results based on the proposed model and discuss their implications on realistic system design. Because radio resources are much more scarce than computational resources, it is reasonable to save as much radio resources as possible, even at the expense of more expenditure on computational resources. Therefore, for a given traffic load $a$ and blocking probability threshold $P_b^{th}$, we first find the minimum number of r-servers $K$ such that (\ref{pb}) holds with sufficient c-servers. After determining $K$, we decrease the number of c-servers from its maximum value ($M K$) down to $0$ and investigate the session blocking probability to determine the maximum statistical multiplexing gain.
    \subsection{Basic Blocking Characteristics}

    \begin{figure}[!t]
    \centering
    \includegraphics[width=3.5in]{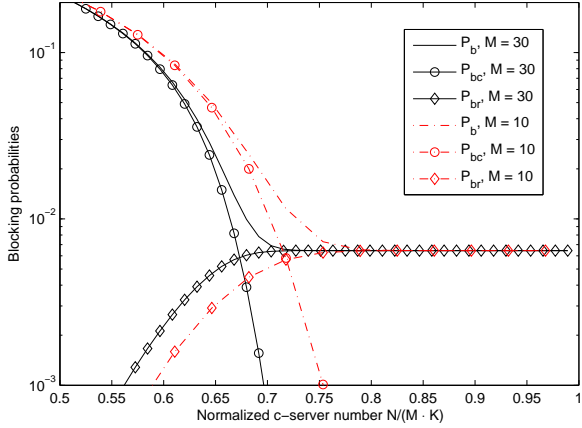}
    \caption{$P_b$, $P_{bc}$, and $P_{br}$ versus normalized c-server number when $M = 10$ and $M = 30$. $a=17.8$, $P_b^{th}=10^{-2}$, and $K=28$.}
    \label{fig1}
    \end{figure}

    Fig. \ref{fig1} shows $P_b$, $P_{br}$, and $P_{bc}$ against different number of c-servers. We see as long as the decrease in $N$ is small and do not pass a certain ``knee point'', $P_b$ will remain almost the same and is dominated by radio blocking. Yet after this ``knee point'', $P_{bc}$ starts to overwhelm $P_{br}$, resulting in an exponential increase in $P_b$. There exist a plain region on the right side of the ``knee point'' because there are only little number of states in $\mathbb{K}_{=N}$ and out of $\mathbb{K}_{<N}^{m,K}$ when $N$ is large, so the magnitude of the resulting product terms in (\ref{pbr}) and (\ref{pbc}) are negligible, i.e. $P_{br} \approx P_b$ and $P_{bc} \approx 0$. This observation implies that we can turn off some c-servers in the {VBS} pool while keeping the {QoS} almost unchanged. This coincides with the statistical multiplexing gain observed in previous work.

    \subsection{Decreasing Marginal Pooling Gain}

    \begin{figure}[!t]
    \centering
    \includegraphics[width=3.5in]{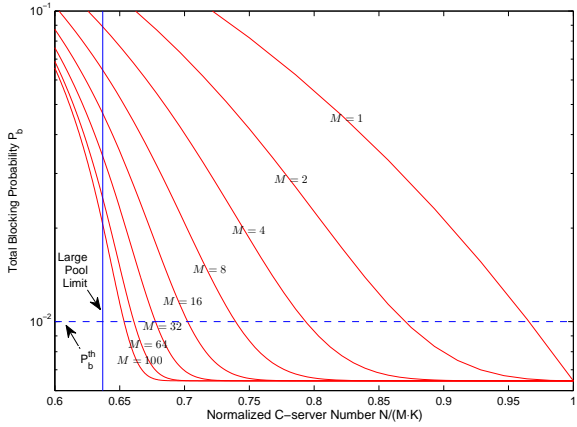}
    \caption{$P_b$ versus normalized c-server number under different pool size $M$. $a=17.8$, $P_b^{th}=10^{-2}$, and $K=28$.}
    \label{fig2}
    \end{figure}

    \begin{figure*}[!t]
    \centerline{\subfloat[$P_b^{th} = 10^{-2}$]{\includegraphics[width=3.5in]{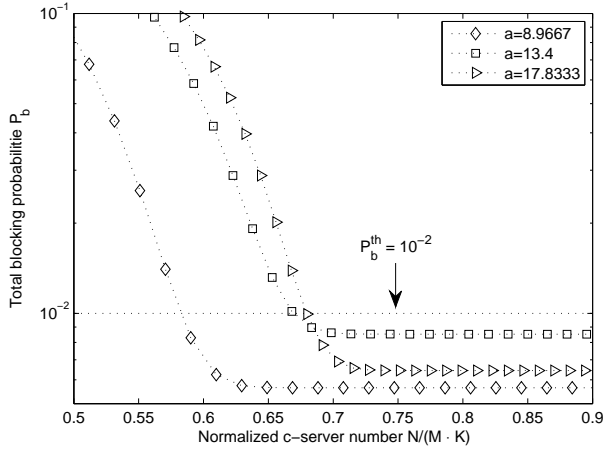}%
    \label{fig3a}}
    \hfil
    \subfloat[$P_b^{th} = 3*10^{-2}$]{\includegraphics[width=3.5in]{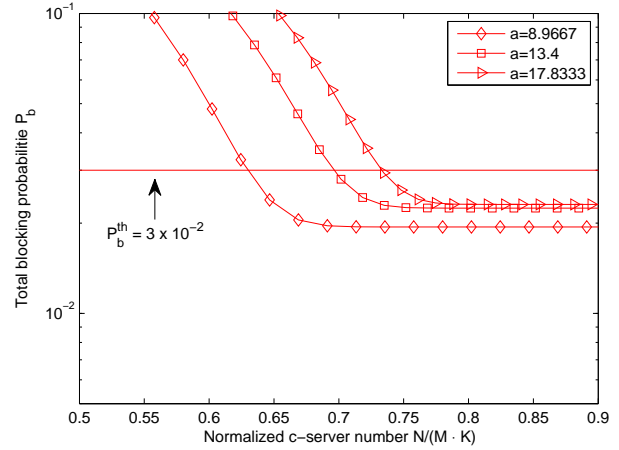}%
    \label{fig3b}}}
    \caption{$P_b$ versus normalized c-server number under different traffic load and {QoS} level. $M=30$.}
    \label{fig3}
    \end{figure*}

    Fig. \ref{fig2} shows the overall blocking probability $P_b$ against varying $N$ under different pool size $M$. According to the discussion in \ref{subsec:large}, we also plot the approximated large pool limit. We find a medium sized {VBS} pool can readily obtain considerable statistical multiplexing gain (the ``knee points'' of $M = 64$ and $M = 100$ are almost the same). Also, the marginal statistical multiplexing gain diminishes fast (even an exponential increase ($2^i$) in $M$ can not keep up with it). Thus a huge number of {VBSs} is needed so that the pooling gain can approach the large pool limit. These observations imply that a {C-RAN} formed with multiple medium sized {VBS} pools can obtain almost the same pooling gain as the one formed with a single huge pool. If we further take the expenditure of fronthaul network into consideration, the former choice may be far more economical than the latter one.

    \subsection{Influence of Traffic Load and {QoS} Level}
    The statistical multiplexing gain of a {VBS} pool also depends on the traffic load and the desired {QoS} level. We illustrate these dependencies in Fig. \ref{fig3}. Firstly we can see that the increase in traffic load will reduce the pooling gain by pushing the ``knee point'' to the right, which is a direct result of the way we reserve radio resources as stated in the beginning of this section. This observation indicates that, in order to get a satisfactory pooling gain under fluctuating traffic load, we may need to dynamically adjust the size of {VBS} pools. Also, by contrasting Fig. \ref{fig3a} and Fig. \ref{fig3b}, we can find that stricter {QoS} requirements can increase the pooling gain. This is because on one hand, we need to increase the number of r-servers $K$ in order to reduce the blocking probability, and this will increase $M K$; On the other hand, the average number of r-servers occupied in a {VBS} is always around $a$. Therefore the stricter the {QoS}, the more idle r-servers and c-servers there will be in the {VBS} pool. Hence when designing realistic {VBS} pools, one may also need to take {QoS} into consideration.

\section{Conclusion and Future Work}
\label{sec:conclusion}
In this article, we propose a multi-dimensional Markov model for {VBS} pools to analyze their statistical multiplexing gain. We showed that the proposed model have a product-form expression for its stationary distribution. Based on this expression, we also derived a recursive method for calculating the blocking probability of a {VBS} pool, which can speed-up evaluation. We further derive the limit of pooling gain as the pool size $M \to \infty,$ and found that the asymptotic pooling gain can be easily approximated as the ratio between traffic load $a$ and the number of r-servers $K$. With numerical results, we found that: 1) the proposed model presents a pooling gain which coincides with the observations in previous work; 2) the pooling gain reaches a significant level even with medium pool size, and the marginal gain of larger pool size is negligible; 3) lighter traffic load and more strict {QoS} level can increase the pooling gain.

Our problem formulation may be extended for cooperative processing senario such as coordinated multipoint transmission ({CoMP}) by allowing a coopative user-session to occupy r-servers of different VBSs and more than one c-server. This will however present more complicated transition properties and need further consideration. Besides, although we have derived a closed-form expression for the stationary distribution, we do not have a closed-form expression for the blocking probability and statistical multiplexing gain, hence their evaluation requires intensive computation. Closed-form expressions for the evaluation of the exact or approximated values call for further research. Also, the proposed model only reflects the statistical multiplexing gain due to session-level randomness, while realistic systems also exhibit certain randomness due to other factors such as the processing of baseband tasks. The significance of other factors require further analysis. What's more, the Markovian property of proposed model is a simplification of realistic {VBS} pools. The accuracy of this assumption awaits the verification based on realistic data. More complex models may be needed to better reflect the mechanisms of realistic {VBS} pools.

\section*{Acknowledgment}
This  work  is  sponsored  in  part  by  the  {National Basic  Research  Program  of China  (973 Program: 2012CB316001}), {the National Science Foundation of China (NSFC) under grant No. 61201191 and No. 61401250}, {the Creative Research Groups of NSFC under grant No. 61321061}, and {Intel Collaborative Research Institute for Mobile Networking and Computing}.

\ifCLASSOPTIONcaptionsoff
  \newpage
\fi

\IEEEtriggeratref{7}
\bibliographystyle{IEEEtran}
\bibliography{IEEEabrv,myref}

\end{document}